\documentclass[11pt]{article} 
\usepackage{fullpage}
\usepackage[font=footnotesize]{caption}
\usepackage{comment}
\usepackage{amsmath}
\usepackage{amsthm}
\usepackage{amssymb}
\usepackage{url}
\usepackage[final]{showkeys}
\usepackage[textsize=scriptsize]{todonotes}
\usepackage{cleveref}
\usepackage{tikz}
\usepackage{lineno}

\newtheorem{theorem}{Theorem}
\newtheorem{lemma}{Lemma}

\newtheorem{definition}{Definition}

\newtheorem{fact}{Fact}
\newtheorem{observation}{Observation}

{\makeatletter
 \gdef\xxxmark{
   \expandafter\ifx\csname @mpargs\endcsname\relax 
     \expandafter\ifx\csname @captype\endcsname\relax 
       \marginpar{xxx}
     \else
       xxx 
     \fi
   \else
     xxx 
   \fi}
 \gdef\xxx{\@ifnextchar[\xxx@lab\xxx@nolab}
 \long\gdef\xxx@lab[#1]#2{{\bf [\xxxmark #2 ---{\sc #1}]}}
 \long\gdef\xxx@nolab#1{{\bf [\xxxmark #1]}}
}

\newcommand{\EDIT}{\mbox{EDIT}}

\newcommand{\CG}{\mbox{CG}}
\newcommand{\VG}{\mbox{VG}}
\newcommand{\PAT}{\mbox{PATTERN}}
\newcommand{\OVT}{\mbox{Orthogonal Vectors Problem}}

\newcommand{\iM}{i \in [d]}

\title{Edit Distance Cannot Be Computed\\ in Strongly Subquadratic Time\\ (unless SETH is false)\footnote{A preliminary version of this paper appeared in \emph{Proceedings of the Forty-Seventh Annual ACM Symposium on Theory of Computing}, 2015.}}
\date{}
\author{
	Arturs Backurs\footnote{\texttt{backurs@mit.edu}}\\ MIT 
	\and Piotr Indyk\footnote{\texttt{indyk@mit.edu}}\\ MIT
	}

\begin{document}

\begin{titlepage}
\clearpage\maketitle
\thispagestyle{empty}
\begin{abstract}
The edit distance (a.k.a.\ the Levenshtein distance) between two strings is defined as the minimum number of insertions, deletions or substitutions of symbols needed to transform one string into another. The problem of computing the edit distance between two strings is a classical computational task, with a well-known algorithm based on dynamic programming. Unfortunately, all known algorithms for this problem run in nearly quadratic time.

In this paper we provide evidence that the near-quadratic running time bounds known for the problem of computing edit distance might be {tight}. Specifically, we show that, if the edit distance can be computed in time $O(n^{2-\delta})$ for some constant $\delta>0$, then the satisfiability of conjunctive normal form formulas with $N$ variables and $M$ clauses can be solved in time $M^{O(1)} 2^{(1-\epsilon)N}$ for a constant $\epsilon>0$. The latter result would violate the {\em Strong Exponential Time Hypothesis}, which postulates that such algorithms do not exist.
\end{abstract}
\end{titlepage}

\section{Introduction}

The edit distance (a.k.a. the Levenshtein distance) between two strings is defined as the minimum number of insertions, deletions or substitutions of symbols needed to transform one string into another. The distance and its generalizations have many applications in computational biology, natural language processing and information theory.  The problem of  computing the edit distance between two strings is a classical computational task, with a well-known algorithm based on the dynamic programming. Unfortunately, that algorithm runs in quadratic time, which is prohibitive for long sequences\footnote{For example, the analysis given in~\cite{estimate} estimates that aligning human and mouse genomes using this approach would take  about 95 years.}.
A considerable effort has been invested into designing faster algorithms, either by assuming that the edit distance is bounded, by considering the average case or by resorting to approximation\footnote{There is a rather large body of work devoted to edit distance algorithms and we will not attempt to list all relevant works here. Instead, we refer the reader to the survey~\cite{navarro2001guided} for a detailed overview of known exact and probabilistic algorithms, and to the recent paper~\cite{andoni2010polylogarithmic} for an overview of approximation algorithms.}. However, the fastest known exact algorithm, due to~\cite{masek1980faster}, has a running time of $O(n^2/\log^2 n)$ for sequences of length $n$, which is still nearly quadratic.

In this paper we provide evidence that the (near)-quadratic running time bounds known for this problem might, in fact, be {tight}. Specifically, we show that if the edit distance can be computed in time $O(n^{2-\delta})$ for some constant $\delta>0$, then the satisfiability of conjunctive normal form (CNF) formulas with $N$ variables and $M$ clauses can be solved in time $M^{O(1)} 2^{(1-\epsilon)N}$ for a constant $\epsilon>0$. The latter result would violate the {\em Strong Exponential Time Hypothesis (SETH)}, introduced in~\cite{impagliazzo2001complexity,impagliazzo2001problems},  which postulates that such algorithms do not exist\footnote{Technically, our results relies on an even weaker conjecture. See Preliminaries for more details.}. The rationale behind this hypothesis is that, despite decades of research on fast algorithms for satisfiability and related problems, no algorithm was yet shown to run in time  faster than $2^{N(1-o(1))}$. Because of this state of affairs, SETH has served as the basis for proving conditional lower bounds for several important computational problems, including  k-Dominating Set~\cite{PW}, the diameter of sparse graphs~\cite{RW},  local alignment~\cite{AWW}, dynamic connectivity problems~\cite{AV}, and the Frechet distance computation~\cite{Bring}.
Our paper builds on these works, identifying a new important  member of the class of ``SETH-hard'' problems. 

\paragraph{Our techniques and related work} 
This work has been stimulated by the recent result of Karl Bringmann~\cite{Bring}, who showed an analogous hardness result for computing the Frechet distance\footnote{Given two sequences of points $P_1$ and $P_2$, the Frechet distance between them is defined as the minimum, over all monotone traversals of $P_1$ and $P_2$, of the largest distance between the corresponding points at any stage of the traversal.},
 and listed SETH-hardness of edit distance as an open problem.  There are notable similarities between  the edit distance and the Frechet distance. In particular, both can be computed in quadratic time, via dynamic programming over an $n \times n$ table $T$ where each entry $T[i,j]$  holds the distance between the first $i$ elements of the first sequence and the first $j$ elements of the second sequence. Furthermore,  in both cases each entry $T[i,j]$ can be computed locally given $T[i,j-1]$, $T[i-1,j]$ and $T[i-1,j-1]$.  The key difference between the two distances is that while the recursive formula for the Frechet distance uses the max function, the formula for the edit distance involves the sum. As a result, the Frechet distance is effectively determined by a single pair of sequence elements, while the edit distance is determined by many pairs of elements. As we describe below, this makes the reduction to edit distance much more subtle.\footnote{This also means that our hardness argument does not extend to the approximate edit distance computation, in contrast to the argument in~\cite{Bring}.}

Our result is obtained by a reduction from the {\em Orthogonal Vectors Problem}, which is defined as follows. Given two sets $A,B\subseteq \{0,1\}^d$ such that $|A|=|B|=N$, the goal is to determine whether there exists $x \in A$ and $y \in B$ such that the dot product  $x \cdot y =\sum_{j=1}^d x_j y_j$ (taken over reals) is equal to $0$.   It is known~\cite{williams2005new} that an $O(d^{O(1)} \cdot N^{2-\delta})$-time algorithm for the Orthogonal Vectors Problem would imply that SETH is false (even in the setting $d=\omega(\log n)$). Therefore, in what follows we focus on reducing the Orthogonal Vectors Problem to the Edit Distance problem.

The first step of our reduction mimics the approaches in~\cite{Bring} (as well as \cite{AWW}). In particular, each $x \in A$ and $y \in B$ is assigned a ``gadget'' sequence. Then, the gadget sequences for all $a \in A$ are concatenated together to form the first input sequence, and the gadget sequences for all $b \in B$ are concatenated to form the second input sequence. 
The correctness of the reduction is proven by showing that:
\begin{itemize}
\item If there is a pair of orthogonal vectors $x \in A$ and $y \in B$,  then one can traverse the two sequences in a way that the gadgets assigned to $x$ and $y$ are {\em aligned}, which implies that the distance induced by this traversal is ``small''.
\item If there is no orthogonal pair, then no such traversal exists, which implies that the distance induced by any traversal is ``large''.
\end{itemize}

The mechanics of this argument depends on the specific distance function. In the case of Frechet distance, the output value is determined by the maximum distance between the aligned elements, so it suffices to show that  the distance between two vector gadgets is smaller than $C$ if they are orthogonal and at least $C$ if they are not, for some value of $C$. In contrast, edit distance {\em sums up} the distances between {\em all} aligned gadgets (as well as the costs of insertions and deletions used to create the alignment), which imposes stronger requirements on the construction. Specifically, we need to show that if two vectors  $x$ and $y$ are not orthogonal, i.e.,  they have at least one overlapping $1$, then the distance between their gadgets is {\em equal} to $C$, not just at least $C$. Since we  need to ensure that the distance between two gadgets cannot grow  in the number of overlapping $1$s, our gadget design and analysis become more complex.

Fortunately, the edit distance is expressive enough to support this functionality. The basic idea behind the gadget construction is to use the fact that the edit distance between two gadget strings, say $VG_1$ (from the first sequence) and $VG_2$ (from the second sequence), is the minimum cost over all possible alignments between $VG_1$ and $VG_2$. 
Specifically, we construct gadgets that allow two alignment options. 
The first option results in a cost that is linear in the number of  overlapping $1$s of the corresponding vectors (this is easily achieved by using substitutions only). On the other hand, the second ``fallback'' option has a fixed cost (say $C_1$) that is slightly higher than the cost of the first option when no $1$s are overlapping (say, $C_0$). Thus,  by taking the minimum of these two options, the resulting  cost is equal to $C_0$ when the vectors are orthogonal and equal to $C_1$ $(>C_0)$ otherwise, which is what is needed. 
See Theorems~\ref{th_min_edit} and~\ref{th_edit_orth} for the details of the construction.

\paragraph{Further developments} Following this work, two recent publications showed multiple results demonstrating conditional hardness of the edit distance, the longest common subsequence problem (LCS), dynamic time warping  problem and other similarity measures between sequences \cite{followup1,followup2}. Among other results,~\cite{followup2} showed hardness of computing the edit distance over the {\em binary} alphabet, which improves over the alphabet size of $7$ required for our reduction. In another development, \cite{followup3} showed that the quadratic hardness of LCS and edit distance computation can be based on a weaker (and therefore more plausible) assumption than SETH, by replacing CNF formulas with more general circuits.  
\section{Preliminaries}

\paragraph{Edit distance}
For any two sequences $x$ and $y$ over an alphabet $\Sigma$, the edit distance $\EDIT(x,y)$ is equal to the minimum number of symbol insertions, symbol deletions or symbol substitutions needed to transform $x$ into $y$. It is well known that the $\EDIT$ function induces a metric; in particular, it is symmetric and satisfies the triangle inequality.

In the remainder of this paper we will use use an equivalent definition of $\EDIT$ that will make the analysis of our reductions more convenient. 

\begin{observation}
\label{no_insertion}
For any two sequences $x,y$, $\EDIT(x,y)$ is equal to the minimum, over all sequences $z$, of the number of deletions and substitutions needed to transform $x$ into $z$ and $y$ into $z$.
\end{observation}
\begin{proof}
It follows directly from the metric properties of the edit distance that $\EDIT(x,y)$ is equal to the minimum, over all sequences $z$, of the number of {\em insertions}, deletions and substitutions needed to transform $x$ into $z$ and $y$ into $z$.
Furthermore, observe that if, while transforming $x$,  we insert a symbol that is later aligned with some symbol of $y$, we can instead delete the corresponding symbol in $y$.
Thus, it suffices to allow deletions and substitutions only.
\end{proof}

\begin{definition}
	\label{pattern}
	We define the following similarity distance between sequences $P_1$ and $P_2$
	and we call it the pattern matching distance between $P_1$ and $P_2$.
	$$\PAT(P_1,P_2)=\min_{\substack{x\text{ is a contiguous}\\
					 \text{ subsequence of }P_2}}\EDIT(P_1,x).$$
\end{definition}

For a symbol $a$ and an integer $i$ we use $a^i$ to denote symbol $a$ repeated $i$ times. 

\paragraph{Orthogonal Vectors Problem} The {\em Orthogonal Vectors Problem} (OVP) is defined as follows: given  two sets $A,B\subseteq \{0,1\}^d$ such that $|A|=|B|=N$, determine whether there exists $x \in A$ and $y \in B$ such that the dot product  $x\cdot y =\sum_{j=1}^d x_j y_j$ (taken over reals) is equal to $0$.   
An alternative formulation of this problem is: given two collections of $N$ sets each, determine if there is a set in the first collection that does not intersect a set from the second collection.\footnote{Equivalently, after complementing sets from the second collection, determine if there is a set in the first collection that is contained in a set from the second collection.}

The Orthogonal Vectors Problem has an easy $O(N^2 d)$-time solution. 
The currently best known algorithm for this problem runs in time $n^{2-1/O(\log c)}$, where $c=d/\log n$~\cite{chan2016deterministic,abboud2015more}.
The Orthogonal Vector Conjecture~\cite{williams2005new,williams2015hardness} postulates that there is no strongly sub-quadratic\footnote{``Strongly sub-quadratic'' means $d^{O(1)}\cdot N^{2-\delta}$ for some constant $\delta>0$.} running time algorithm for OVP.
Moreover, it is known that any algorithm for this problem with strongly sub-quadratic running time would also yield a more efficient algorithm for CNF-SAT, breaking SETH~\cite{williams2005new}. Thus, in what follows, we focus on reducing the Orthogonal Vectors Problem to $\EDIT$.

\paragraph{Simplifying assumption} We assume that in the Orthogonal Vectors Problem, for all vectors $b \in B$, $b_1=1$, that is, the first coordinate of any vector $b \in B$ is equal to $1$. We can make this assumption w.l.o.g. because we can always add a $1$ to the beginning of each $b \in B$, and add a $0$ to the beginning of each $a \in A$.
 
\section{Reductions}

\subsection{Vector gadgets} \label{vc}
We now describe vector gadgets as well as provide some intuition behind the construction.

We will construct sequences over an alphabet $\Sigma=\{0,1,2,3,4\}$. 

We start by defining an integer parameter $l_0=1000 \cdot d$, where $d$ is the dimensionality of the vectors in the Orthogonal Vectors Problem.
We then define {\em coordinate gadget} sequences $\CG_1$ and $\CG_2$ as follows.
For integer $x \in \{0,1\}$ we define
$$
	\CG_1(x):=
	\begin{cases}
		2^{l_0}\,0\,1\,1\,1\,2^{l_0} & \text{if } x=0; \\
		2^{l_0}\,0\,0\,0\,1\,2^{l_0} & \text{if } x=1,
	\end{cases}
$$
$$
	\CG_2(x):=
	\begin{cases}
		2^{l_0}\,0\,0\,1\,1\,2^{l_0} & \text{if } x=0; \\
		2^{l_0}\,1\,1\,1\,1\,2^{l_0} & \text{if } x=1.
	\end{cases}
$$

The coordinate gadgets were designed so that they have the following properties.
For any two integers $x_1,x_2 \in \{0,1\}$, 
$$
	\EDIT(\CG_1(x_1),\CG_2(x_2))=
	\begin{cases}
		1 & \text{if } x_1 \cdot x_2=0; \\
		3 & \text{if } x_1 \cdot x_2=1.
	\end{cases}
$$

Further, we define another parameter $l_1=(1000 \cdot d)^2$. 
We use $\Sigma$-style notation to denote the concatenation of sequences. For example, given $d$ sequences $s_1, \ldots, s_d$, we denote the concatenation $s_1 \ldots s_d$ by $\bigcirc_{\iM}s_i$.
For vectors $a,a',b\in \{0,1\}^d$, we define the {\em vector gadget} sequences as
$$
	\VG_1(a,a')=Z_1 L(a) V_0 R(a') Z_2\text{ and }\VG_2(b)=V_1 D(b) V_2,
$$ 
where
$$
	V_1=V_2=V_0=3^{l_1}, \ \ Z_1=Z_2=4^{l_1},
$$ 
$$
	L(a)=\bigcirc_{\iM}\CG_1(a_i), \ \ R(a')=\bigcirc_{\iM}\CG_1(a'_i), \ \ D(b)=\bigcirc_{\iM}\CG_2(b_i).
$$ 
In what follows we skip the arguments of $L$, $R$ and $D$.
We denote the length of $L$, $R$ and $D$ by 
$l=|L|=|R|=|D|=d(4+2l_0)$.

We visualize the defined vector gadgets in Figure \ref{Vector_Gadgets}.

\paragraph{Intuition behind the construction} Before going into the analysis of the gadgets in Section \ref{gadget_properties}, we will first provide some intuition behind the construction. Given three vectors $a,a',b\in\{0,1\}^d$, we want that $\EDIT(\VG_1(a,a'),\VG_2(b))$ grows linearly in the minimum of $a \cdot b$ and $a' \cdot b$.
More precisely, we want that 
\begin{equation} \label{prop}
	\EDIT(\VG_1(a,a'),\VG_2(b))=C+t \cdot \min(a \cdot b,a' \cdot b),
\end{equation}
where the integers $C,t>0$ are functions of $d$ only. In fact, we will have that $t=2$. To realize this, we construct our vector gadgets $\VG_1$ and $\VG_2$ such that there are only two possibilities to achieve small edit distance. In the first case, the edit distance grows linearly in $a \cdot b$. In the second case, the edit distance grows linearly in $a' \cdot b$. Because the edit distance is equal to the minimum over all possible alignments, we take the minimum of the two inner products. After taking the minimum, the edit distance will satisfy the properties stated in (1). More precisely, we achieve the minimum edit distance cost between $\VG_1$ and $\VG_2$ by following one of the following two possible sequences of operations:
\begin{itemize}
	\item Case 1: Delete $Z_1$ and $L$. Substitute $Z_2$ with $V_2$. This costs $C':=l_1+d\cdot (2l_0+4)+l_1$. Transform $R$ and $D$ into the same sequence by transforming the corresponding coordinate gadgets into the same sequences. By the construction of the coordinate gadgets, the cost of this step is $d+2\cdot(a'\cdot b)$. Therefore, this case corresponds to edit distance cost $C'+d+2\cdot(a'\cdot b)=C+2\cdot(a'\cdot b)$.
	\item Case 2: Delete $R$ and $Z_2$. Substitute $Z_1$ with $V_1$. This costs $C'$. Transform $L$ and $D$ into the same sequence by transforming the corresponding coordinate gadgets. Similarly as before, the cost of this step is $d+2\cdot(a \cdot b)$. Therefore, this case corresponds to edit distance cost $C'+d+2\cdot(a\cdot b)=C+2\cdot(a\cdot b)$.
\end{itemize}
Taking  the minimum of these two cases yields the desired formula \eqref{prop}. 

In the reduction given in  Section \ref{reductionpat},  we ensure that the dot product $a' \cdot b$ is always equal to $1$. 
As a result we have that  $\EDIT(\VG_1(a),\VG_2(b))$ is small (equal to $C_0$) if the vectors $a$ and $b$ are orthogonal,  and is large (equal to $C_1$) otherwise. That is: 
\begin{equation} \label{prop2}
	\EDIT(\VG_1(a),\VG_2(b))=
	\begin{cases}
		C_0 & \text{if }a\cdot b=0\\
		C_1 & \text{otherwise}
	\end{cases}
\end{equation}
for $C_1>C_0$. This property is crucial for our construction, as it guarantees that the sum of several terms $\EDIT(\VG_1(a),\VG_2(b))$ is smaller than some threshold  if and only if $a \cdot b=0$ for at least one pair of vectors $a$ and $b$.
This enables us to detect whether such a pair exists. 
In contrast, this property would not hold if $\EDIT(\VG_1(a),\VG_2(b))$ depended linearly on the value of $a \cdot b$.

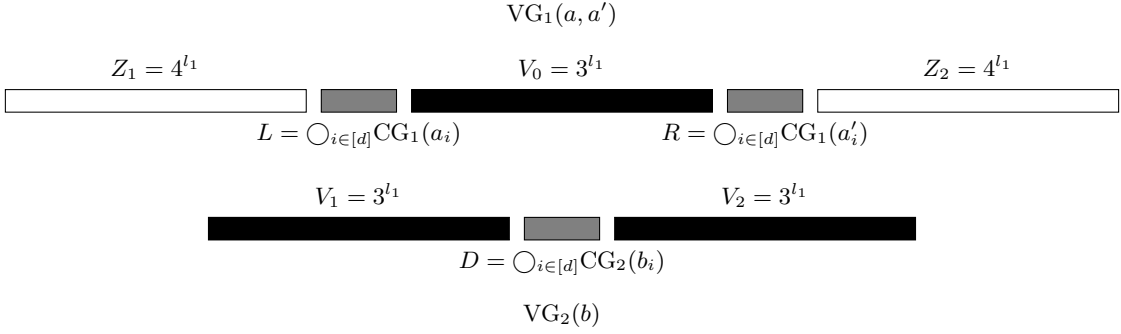
\begin{figure*}
        \centering
        \begin{tikzpicture} [solid][font=\footnotesize]
		\fill[white,draw=black] (-7.4,1) rectangle (-3.4,0.7);
		\fill[gray,draw=black] (-3.2,1) rectangle (-2.2,0.7);
		\fill[black,draw=black] (-2,1) rectangle (2,0.7);
		\fill[gray,draw=black] (3.2,1) rectangle (2.2,0.7);
		\fill[white,draw=black] (7.4,1) rectangle (3.4,0.7);
		
		\fill[black,draw=black] (-4.7,-1) rectangle (-0.7,-0.7);
		\fill[gray,draw=black] (-0.5,-1) rectangle (0.5,-0.7);
		\fill[black,draw=black] (4.7,-1) rectangle (0.7,-0.7);

		\node (s) at (-5.4,1.3) {$Z_1=4^{l_1}$};
		\node (s) at (-2.7,0.4) {$L=\bigcirc_{\iM}\CG_1(a_i)$};
		\node (s) at (0,1.3) {$V_0=3^{l_1}$};
		\node (s) at (2.7,0.4) {$R=\bigcirc_{\iM}\CG_1(a'_i)$};
		\node (s) at (5.4,1.3) {$Z_2=4^{l_1}$};
		\node (s) at (-2.7,-0.4) {$V_1=3^{l_1}$};
		\node (s) at (2.7,-0.4) {$V_2=3^{l_1}$};
		\node (s) at (0,-1.3) {$D=\bigcirc_{\iM}\CG_2(b_i)$};
		
		\node (s) at (0,2) {$\VG_1(a,a')$};
		\node (s) at (0,-2) {$\VG_2(b)$};
		
        \end{tikzpicture}
        \caption{\footnotesize A visualisation of the vector gadgets.
		A black rectangle denotes a run of $3$s, while a white rectangle denotes a run of $4$s. 
		A gray rectangle denotes a sequence that contains $0$s, $1$s and $2$s. A short rectangle denotes
		a sequence of length $l$, while a long one denotes a sequence of length $l_1$.} 
	\label{Vector_Gadgets}
\end{figure*}

\subsubsection{Properties of the vector gadgets} \label{gadget_properties}

\begin{theorem} \label{th_min_edit}
	For any vectors $a,a',b \in \{0,1\}^d$,
	\begin{equation*} \label{min_edit}
		\EDIT(\VG_1(a,a'),\VG_2(b))=2l_1+l+d+2\cdot\min\left(a\cdot b,a'\cdot b\right).
	\end{equation*}
\end{theorem}
\begin{proof} Follows from lemmas \ref{upper_bound} and \ref{lower_bound} below.
\end{proof}

\begin{lemma} \label{upper_bound}
	For any vectors $a,a',b \in \{0,1\}^d$,
	\begin{equation*}
		\EDIT(\VG_1(a,a'),\VG_2(b))\leq 2l_1+l+d+2\cdot\min\left(a\cdot b,a'\cdot b\right).
	\end{equation*}
\end{lemma}
\begin{proof}
	W.l.o.g., $a\cdot b\leq a' \cdot b$. We delete $R$ and $Z_2$ from $\VG_1(a,a')$. This costs $l_1+l$.
	We transform $Z_1LV_0$ into $V_1DV_2$ by using substitutions only. This costs $l_1+d+2\cdot(a\cdot b)$. We get the upper bound on the $\EDIT$ cost and this finishes the proof.
\end{proof}

\begin{lemma} \label{lower_bound}
	For any vectors $a,a',b \in \{0,1\}^d$,
	\begin{equation*}
		\EDIT(\VG_1(a,a'),\VG_2(b))\geq 2l_1+l+d+2\cdot\min\left(a\cdot b,a'\cdot b\right)=:X.
	\end{equation*}
\end{lemma}
\begin{proof}
	Consider an optimal transformation of $\VG_1(a,a')$ and $\VG_2(b)$ into 
	the same sequence. Every symbol (say $s$)  in the first sequence is either
	substituted, preserved or deleted in the process. If a symbol is not
	deleted but instead is preserved or substituted by another symbol (say $t$), 
	we say that $s$ is {\em aligned} with $t$, or that $s$ and $t$ have an  {\em alignment}.

	We state the following fact without a proof.

	\begin{fact}
		Suppose we have two sequences $x$ and $y$ of symbols.
		Let $i_1<j_1$ and $i_2<j_2$ be four positive integers.
		If $x_{i_1}$ is aligned with $y_{j_2}$, then
		$x_{j_1}$ cannot be aligned with $y_{i_2}$.
	\end{fact} 
	
	From now on we proceed by considering three cases.

	{\bf Case 1.} The subsequence $D$ has alignments with both $Z_1L$ \emph{and} $RZ_2$. 
	In this case, the cost induced by symbols from $Z_1$ and $Z_2$, and $V_0$ is
	$l_1$ for each one of these sequences because the symbols must be deleted or substituted.
	This implies that $\EDIT(\VG_1(a),\VG_2(b))\geq 3l_1$,
	which contradicts an easy upper bound. 
	We have an upper bound $\EDIT(\VG_1(a),\VG_2(b))\leq 2l_1+3l$, which is obtained by 
	deleting $L$, $R$, $D$, $Z_1$ and replacing $Z_2$ with $V_2$ symbol by symbol.
	Remember that $l_0=1000\cdot d$ and $l_1=(1000\cdot d)^2$, and $l=d(4+l_0)$. Thus, $l_1\geq 3l$ and the lower bounds contradicts the upper bound.
	Therefore, this case cannot occur.

	{\bf Case 2.} $D$ does not have any alignments with $Z_1L$. We will show that, if this case happens, then
	$$
		\EDIT(\VG_1(a,a'),\VG_2(b))\geq 2l_1+l+d+2\cdot\left(a'\cdot b\right).
	$$
	
	We start by introducing the following notion. Let $v$ and $z$ be two sequences that decompose as $v=x V $ and $z=y Z$. Consider two sequences ${\cal T}$ and ${\cal R}$ of deletions and substitutions that transform $v$ into $u$ and $z$ into $u$, respectively. An operation in ${\cal T}$ or ${\cal R}$ is called {\em internal to} $V$ and $Z$ if it is either a (1) deletion of a symbol in $V$ or $Z$, or (2) a substitution of a symbol in $V$ so that it aligns with a symbol in $Z$, or vice versa. All other operations, including substitutions that align with  symbols in $V$ ($Z$, resp.) to those outside of $Z$ ($V$, resp.) are called {\em external to} $V$ and $Z$.

	We state the following  fact without a proof.
	\begin{fact}
		\label{suffixes}
		Let $xV$ and $yZ$ be sequences such that $V=4^t$, $Z=3^t$ and
		$x$ and $y$ are arbitrary sequences over an arbitrary alphabet not including $3$ or $4$.
		Consider $\EDIT(xV,yZ)$ and the corresponding operations minimizing the distance.
		Among those operations, the number of operations that are internal to $V$ and $Z$ is at least $t$.
	\end{fact}
	
	Given that $|Z_2|=|V_2|=l_1$ and $Z_2$ consists only of $4$s
	and $V_2$ consists of only $3$s, Fact \ref{suffixes} 
	implies that the number of operations that are internal to $Z_2$ and $V_2$ is at least $S_1:=l_1$.
	
	Because $D$ does not have any alignments with $Z_1L$, we must have that every symbol in $Z_1L$ gets deleted or substituted.
	Thus, the total contribution from symbols in $Z_1L$ to an optimal alignment is $S_2:=|Z_1L|=l_1+l$. Now we will lower bound the contribution to an optimal alignment from symbols in sequences $R$ and $D$. First, observe that both $R$ and $D$ have $d$ runs of $1$s. We consider the following two sub-cases.
	
	{\bf Case 2.1.}	There exist $i,j \in [d]$ with $i \neq j$ such that the $i$th run in $D$ has alignments with the $j$th run in $R$. The number of symbols of type $2$ to the right of the $i$th run in $D$ and the number of symbols of type $2$ to the right of the $j$th run in $R$ differ by at least $2l_0$. Therefore, the induced $\EDIT$ cost of symbols of type $2$ in $R$ and $D$ is at least 
	$2l_0\geq d+2\cdot\left(a'\cdot b\right)=:S_{3}$, from which we conclude that
	\begin{align*}
		&\EDIT(\VG_1(a,a'),\VG_2(b))\geq S_{1}+S_2+S_3\\
		= \ & l_1+(l_1+l)+\left(d+2\cdot\left(a'\cdot b\right)\right)=X.
	\end{align*}
	
	In the inequality we used the fact that the contributions from $S_1$, $S_2$ and $S_3$ are disjoint. This follows from the definitions of the quantities. 
	More precisely, the contribution from $S_1$ comes from operations that are \emph{internal} to $V_2$ and $Z_2$. Thus, it remains to show that the contributions from $S_2$ and $S_3$ are disjoint. This follows from the fact that the contribution from $S_2$ comes from symbols $Z_1L$ and the assumption that $D$ does not have any alignments with $Z_1L$.

	{\bf Case 2.2.} (The complement of Case 2.1.) Consider any $i \in [d]$.
	If a symbol of type $1$ from the $i$th run in $D$ is aligned with a symbol of type $1$ in $R$, then the symbol of type $1$ comes from the $i$th run in $R$.
	Define the set $P$ as the set of all numbers $i \in [d]$ such that the $i$th run of $1$s in $D$ has alignment with the $i$th run of $1$s in $R$.
	
	For all $i\in P$, the $i$th run in $R$ aligns with 
	the $i$th run in $D$.
	By the construction of coordinate gadgets, the $i$th run in $R$ and $D$ incur
	$\EDIT$ cost $\geq 1+2a'_ib_i$.
	
	For all $i \not\in P$, the $i$th run in $D$ incurs $\EDIT$ cost at least $2$ (since there are at least two symbols of type $1$).
	Similarly, the $i$th run in $R$ incurs $\EDIT$ cost at least $1$ (since there is at least one symbol of type $1$).
	Therefore, for every $i \not\in P$, the $i$th run in $R$ and $D$ incur $\EDIT$ cost $\geq 1+2 \geq 1+2a'_ib_i$.
	
	We get that the total contribution to the $\EDIT$ cost from the $d$ runs in $D$ and the $d$ runs in $R$ is
	$$
		\sum_{i\in P}\left(1+2a'_ib_i\right)+\sum_{i\in [d]\setminus P}3\geq\sum_{i=1}^d\left(1+2a'_ib_i\right)=d+2\cdot\left(a'\cdot b\right)=:S_4.
	$$
	We conclude:
	\begin{align*}
		&\EDIT(\VG_1(a,a'),\VG_2(b))\geq S_{1}+S_2+S_4\\
		= \ & l_1+(l_1+l)+\left(d+2\cdot\left(a'\cdot b\right)\right)=X.
	\end{align*}
	
	We used the fact that the contributions from $S_1$, $S_2$ and $S_4$ are disjoint. The argument is analogous as in the previous case.
	
	{\bf Case 3.} The symbols of $D$ are not aligned with any symbols in $RZ_2$. If this case happens, then
	$$
		\EDIT(\VG_1(a,a'),\VG_2(b))\geq 2l_1+l+d+2\cdot\left(a\cdot b\right).
	$$
	The analysis of this case is analogous to the analysis of Case 2. More concretely, for any sequence $x$, define $\text{reverse}(x)$ to be the sequence $y$ of length $|x|$ such that $y_i=x_{|x|+1-i}$ for all $i=1,2,\ldots,|x|$. Now we repeat the proof in Case 2 but for 
	$$
		\EDIT(\text{reverse}(\VG_1(a,a')),\text{reverse}(\VG_2(b))).
	$$
	This yields exactly the lower bound that we need.
	
	The proof of the lemma follows. We showed that Case 1 cannot happen. By combining lower bounds corresponding to Cases 2 and 3, we get the lower bound stated in the lemma.
\end{proof}

We set $a':=1\,0^{d-1}$, that is, $a'$ is a binary vector of length $d$ such that $a'_1=1$ and $a'_i=0$ for $i=2,\ldots,d$. We define
$$
	VG_1(a):=VG_1(a,a').
$$

\begin{theorem} \label{th_edit_orth}
	Let $a\in \{0,1\}^d$ be any binary vector and $b \in \{0,1\}^d$ be any binary vector that starts with $1$, that is, $b_1=1$. Then,
	\begin{equation*} \label{edit_orth}
		\EDIT(\VG_1(a),\VG_2(b))=
		\begin{cases}
			E_s:=2l_1+l+d & \text{if }a\cdot b=0; \\
			E_u:=2l_1+l+d+2 & \text{if }a\cdot b\geq 1.
		\end{cases}
	\end{equation*}
\end{theorem}
\begin{proof}
	Follows from Theorem \ref{th_min_edit} by setting $a'=1\,0^{d-1}$ and observing that $a' \cdot b=1$ because $b_1=1$.
\end{proof}

\subsection{Reducing the $\OVT$ to $\PAT$} \label{reductionpat}
We proceed by concatenating vector gadgets into sequences. 

We note that the length of the vector gadgets produced by $\VG_1$ depends on the dimensionality $d$ of the vectors but not on the entries of the vectors.
The same is true about $\VG_2$. We set $t$ to be the maximum of the two lengths. Furthermore, we set $T=1000d\cdot t=\Theta(d^3)$.
We define $\VG_k'(a)=5^{T}\VG_k(a)5^{T}$ for $k\in \{1,2\}$.
Let $f=1^d$ be a vector consisting of $d$ entries equal to $1$.

Let $A$ and $B$ be sets from the $\text{Orthogonal Vectors}$ instance. By definition $|A|=|B|$.

We define sequences
$$
	P_1=\bigcirc_{a\in A}\VG_1'(a),
$$
$$
	P_2=\left(\bigcirc_{i=1}^{|A|-1}\VG_2'(f)\right)
		\left(\bigcirc_{b\in B}\VG_2'(b)\right)
		\left(\bigcirc_{i=1}^{|A|-1}\VG_2'(f)\right).
$$

\begin{theorem}
	Let $X:=|A|\cdot E_u$.
	If there are two orthogonal vectors, one from set $A$, another from set $B$, then $\PAT(P_1,P_2)\leq X-(E_u-E_s)$; otherwise we have $\PAT(P_1,P_2)= X$. 
\end{theorem}
\begin{proof}
	Follows from Lemmas \ref{pat_sat} and \ref{pat_unsat} below.
\end{proof}

\begin{lemma}
	\label{pat_sat}
	If there are two orthogonal vectors, one from $A$, another from $B$, then
	$$
		\PAT(P_1,P_2)\leq X-(E_u-E_s)=X-2.
	$$
\end{lemma}
\begin{proof}
	Let $a \in A$ and $b \in B$ be vectors such that $a \cdot b=0$.

	We can choose a contiguous subsequence $s$ of $P_2$ consisting of
	a sequence of $|A|$ vector gadgets $\VG_2'$ 
	such that $s$ has the following property: 
	transforming the vector gadgets $\VG_1'$ from $P_1$ and their corresponding
	vector gadgets $\VG_2'$ from $s$ into the same sequence 
	one by one as per Theorem \ref{th_edit_orth},
	we achieve a cost smaller than the upper bound. 
	We use the fact that at least one
	transformation is cheap because $a \cdot b=0$
	and we choose $s$ so that $\VG_1'(a)$ and $\VG_2'(b)$
	get transformed into the same sequence.
\end{proof}

\begin{lemma}
	\label{pat_unsat}
	If there are no two orthogonal vectors, one from $A$, another from $B$, then $$\PAT(P_1,P_2)=X.$$
\end{lemma}
\begin{proof}
	Consider a graph $(X_1 \cup X_2,E)$ with vertices
	$x_1(a) \in X_1$, $a \in A$, $x_2(b) \in X_2$, $b \in B$. We also add $2|A|-2$ copies of $x_2(f)$ to set $X_2$ corresponding to $2|A|-2$ vectors $f$ in sequence $P_2$.
	Consider an optimal transformation of $P_1$ and a subsequence of $P_2$
	into the same sequence according to Definition \ref{pattern}.
	We connect two vertices $x_1(a)$ and $x_2(b)$
	if and only if $\VG_1(a)$ and $\VG_2(b)$ have an alignment in the transformation.

	We want to claim that every vector gadget $\VG_1(a)$ from $P_1$ contributes
	a cost of at least $E_u$ to the final cost of $\PAT(P_1,P_2)$. This will 
	give $\PAT(P_1,P_2)\geq X$. We consider the connected components of the graph.
	We will show that a connected component that has $r\geq 1$ vertices from $X_1$, contributes
	$\geq r \cdot E_u$ to the final cost of $\PAT(P_1,P_2)$. From the case analysis below we will
	see that these contributions for different connected components are separate. Therefore, 
	by summing up the contributions for all the connected components, we get
	$\PAT(P_1,P_2)\geq |A|\cdot E_u=X$.

	Consider a connected component of the graph with at least one vertex from $X_1$. 
	We examine several cases.

	{\bf Case 1.} The connected component has only one
	vertex from $X_1$. Let $x_1(a)$ be the vertex. 

	{\bf Case 1.1.} $x_1(a)$ is connected to more 
	than one vertex. In this case, $\VG_1(a)$ induces a cost of at least $2T>E_u$
	(this cost is induced by symbols of type $5$).
	
	{\bf Case 1.2.} $x_1(a)$ (corresponding to vector gadget $\VG_1(a)$) 
	is connected to only one vertex $x_2(b)$ (corresponding to vector gadget $\VG_2(b)$).
	Let $x$ be a contiguous substring of $P_2$ that achieves 
	the minimum of $\EDIT(P_1,x)$ (see Definition \ref{pattern}).

	{\bf Case 1.2.1.} The vector gadget 
	$\VG_2(b)$ is fully contained in the substring $x$. 
	We claim that the contribution from symbols in the sequences $\VG_1(a)$ and $\VG_2(b)$ is at least $\EDIT(\VG_1(a),\VG_2(b))$. This is sufficient because we know that $\EDIT(\VG_1(a),\VG_2(b)) \geq E_u$ from Theorem \ref{th_edit_orth}.
	If no symbol in $\VG_1(a)$ or $\VG_2(b)$ is aligned with a symbol of type $5$, the claim follows directly by applying Theorem \ref{th_edit_orth}.
	Otherwise, every symbol that is aligned with a symbol of type $5$ contributes cost $1$ to the final cost. The contribution from symbols in the sequences $\VG_1(a)$ and $\VG_2(b)$ is at least $\EDIT(\VG_1(a),\VG_2(b))$ because we can transform the sequences $\VG_1(a)$ and $\VG_2(b)$ into the same sequence by first deleting the symbols that are aligned with symbols of type $5$ (every such alignment contributes cost $1$) and then transforming the remainders of the sequences $\VG_1(a)$ and $\VG_2(b)$ into the same sequence.

	{\bf Case 1.2.2.} The complement of Case 1.2.1.
	We need to consider this case because of the following reason. We could potentially
	achieve a contribution of $\VG_1(a)$ to $\PAT(P_1,P_2)$ that is smaller than $E_u$ 
	by transforming $\VG_1(a)$ and a \emph{contiguous
	substring} of $\VG_2(b)$ into the same string (instead of transforming
	$\VG_1(a)$ and $\VG_2(b)$ into the same string). In the next
	paragraph we show that this cannot happen.

	$\VG_2(b)$ shares symbols with $x$ and is not fully contained in $x$. 
	$\VG_2(b)$ must be the left-most (right-most, resp.) vector
	gadget in $x$ but then $T$ left-most (right-most, resp.) symbols of type $5$ 
	of $\VG'_1(a)$ induce a cost of at least $T>E_u$ since the symbols of type $5$ cannot be preserved and must be substituted or deleted.

	{\bf Case 1.3.} $x_1(a)$ is connected to no vertex. We get that
	$\VG_1(a)$ induces cost of at least $|\VG_1(a)|>E_u$.

	{\bf Case 2.} The connected component has $r>1$
	vertices $x_1(a)$ from $X_1$. In this case, the cost
	induced by the vector gadgets $\VG_1(a)$ corresponding to the vertices from $X_1$ 
	in the connected component is at least $(r-1)\cdot 2T>r\cdot E_u$ (this cost is induced by symbols of type $2$).
	
	This finishes the argument that $\PAT(P_1,P_2)\geq X$.

	It remains to argue that we can achieve cost $X$ (to show that $\PAT(P_1,P_2)\leq X$)
	and it can be done  as in Lemma \ref{pat_sat}.
\end{proof}

\subsection{Reducing $\PAT$ to $\EDIT$}
 \label{section_edit}
\label{ss:edit}

We set $P_2':=P_2$ and $P_1':=6^{|P_2'|}P_16^{|P_2'|}$. Remember that $|A|=|B|$.

\begin{theorem}
	\label{hardness_bound}
	Let  $Y:=2\cdot |P_2'|+|A|\cdot E_u$. If there are no two orthogonal vectors,
	then $\EDIT(P_1',P_2')=Y$; otherwise  $\EDIT(P_1',P_2')\leq Y-(E_u-E_s)=Y-2$.
\end{theorem}
\begin{proof} Follows from Lemmas \ref{edit_sat} and \ref{edit_unsat} below. \end{proof}

\begin{lemma}
	\label{edit_sat}
	If there are two orthogonal vectors, then $$\EDIT(P_1',P_2')\leq Y-(E_u-E_s)=Y-2.$$
\end{lemma}
\begin{proof}
	We transform $P_1$ and a subsequence of $P_2'$ into the same sequence as in Lemma \ref{pat_sat}.
	We replace the remaining prefix and suffix of $P_2'$ with symbols of type $6$ and
	delete the excess of symbols of type $6$ from $P_1'$.
\end{proof}

\begin{lemma}
	\label{edit_unsat}
	If there are no two orthogonal vectors, then $$\EDIT(P_1',P_2')=Y.$$
\end{lemma}
\begin{proof}
	We can easily check that $\EDIT(P_1',P_2')\leq Y$ as in Lemma \ref{edit_sat}. It remains to prove the opposite inequality.

	$P_1'$ contains $2|P_2'|$ symbols of type $6$. Those will incur
	a cost of at least $2|P_2'|$. $P_1'$ has the remaining subsequence
	$P_1$, which will incur cost at least $\PAT(P_1,P_2')$.
	Using Lemma \ref{pat_unsat}, we finish the proof.
\end{proof}

As a result, we get the following theorem.
\begin{theorem}
	If $\EDIT$ can be computed in time 
	$O(n^{2-\delta})$ for some $\delta>0$
	on two sequences of length $n$ over an 
	alphabet of size $7$, then the $\OVT$ with $|A|=|B|=N$ and $A,B\subseteq \{0,1\}^d$
	can be solved in time $d^{O(1)}\cdot N^{2-\delta}$.
\end{theorem}
\begin{proof}
The proof follows immediately from Theorem~\ref{hardness_bound}.
\end{proof} 
\section{Acknowledgments}
The authors thank Amir Abboud, Karl Bringmann, Sepideh Mahabadi, Ludwig Schmidt and the reviewers for providing helpful comments. This work was supported by an IBM PhD Fellowship, grants from the NSF, the MADALGO center, and the Simons Investigator award.
 
	\bibliographystyle{alpha}

\end{document}